\documentclass[AMA,STIX1COL]{WileyNJD-v2}

\usepackage[utf8]{inputenc}
\usepackage{mathtools}
\usepackage{xspace}
\usepackage{bm}
\usepackage[rightcaption]{sidecap}
\sidecaptionvpos{figure}{t}

\DeclarePairedDelimiter{\norm}{\lVert}{\rVert}
\DeclarePairedDelimiter{\abs}{\lvert}{\rvert}
\DeclareMathOperator{\Exp}{E}
\DeclareMathOperator{\bExp}{\overline{E}}
\newcommand{\sumN}{\sum_{t=1}^N}

\newcommand{\inv}{^{-1}\xspace}
\newcommand{\z}{q\xspace}
\newcommand{\p}{\bm\rho\xspace}

\newcommand{\tran}{^\mathsf{T}\xspace}

\newcommand{\dd}{\,\mathrm{d}\xspace}

\newcommand{\Gq}{G(\z)\xspace}
\newcommand{\Gqi}{G\inv(\z)\xspace}
\newcommand{\Cqp}{C(\z,\p)\xspace}
\newcommand{\Ciqp}{C_\text{i}(\z,\p)\xspace}
\newcommand{\Cfq}{C_\text{f}(\z)\xspace}
\newcommand{\Coq}{C_0(\z)\xspace}
\newcommand{\Tqp}{T(\z,\p)\xspace}
\newcommand{\Toq}{T_0(\z)\xspace}
\newcommand{\Qqp}{Q(\z,\p)\xspace}
\newcommand{\Qqps}{Q(\z,\pstar)\xspace}
\newcommand{\Qiqp}{Q\inv(\z,\p)\xspace}
\newcommand{\Sqp}{S(\z,\p)\xspace}
\newcommand{\Soq}{S_0(\z)\xspace}
\newcommand{\Aqp}{A(\z,\p)\xspace}
\newcommand{\Bqp}{B(\z,\p)\xspace}
\newcommand{\Kqp}{K(\z,\p)\xspace}

\newcommand{\w}{e^{j\omega}\xspace}
\newcommand{\Awp}{A(\w,\p)\xspace}
\newcommand{\Kwp}{K(\w,\p)\xspace}
\newcommand{\Qwp}{Q(\w,\p)\xspace}
\newcommand{\Qiwp}{Q\inv(\w,\p)\xspace}
\newcommand{\Qdw}{Q_\text{d}(\w)\xspace}
\newcommand{\Qdiw}{Q_\text{d}\inv(\w)\xspace}
\newcommand{\Phiy}{\Phi_y(\omega)\xspace}
\newcommand{\Phid}{\Phi_d(\omega)\xspace}
\newcommand{\Phixy}{\Phi_{x,y}(\omega)\xspace}

\newcommand{\Gw}{G(\w)\xspace}
\newcommand{\Tow}{T_0(\w)\xspace}
\newcommand{\Rt}{r(t)\xspace}
\newcommand{\Dt}{d(t)\xspace}
\newcommand{\Ut}{u(t)\xspace}
\newcommand{\Yt}{y(t)\xspace}
\newcommand{\Vt}{v(t)\xspace}

\newcommand{\Ytp}{y(t,\p)\xspace}
\newcommand{\Etp}{e(t,\p)\xspace}
\newcommand{\Utp}{u(t,\p)\xspace}

\newcommand{\Etv}{\bar{e}(t)\xspace}
\newcommand{\efv}{\bar{e}_\text{f}\xspace}
\newcommand{\Efv}{\efv(t)\xspace}
\newcommand{\utv}{\bar{u}\xspace}
\newcommand{\Utv}{\utv(t)\xspace}
\newcommand{\Dtv}{\bar{d}(t)\xspace}
\newcommand{\Rtv}{\bar{r}(t)\xspace}

\newcommand{\Qdq}{Q_\text{d}(\z)\xspace}
\newcommand{\Qdiq}{Q_\text{d}\inv(\z)\xspace}

\newcommand{\Yd}{y_\text{d}(t)\xspace}
\newcommand{\pd}{\p_\text{d}\xspace}
\newcommand{\pstar}{\p^\star\xspace}
\newcommand{\Cdq}{C_\text{d}(\z)\xspace}
\newcommand{\Aqpd}{A(\z,\pd)\xspace}
\newcommand{\Bqpd}{B(\z,\pd)\xspace}
\newcommand{\Vcd}{V_\text{c}(\z,\pd)}
\newcommand{\Vped}{V_\text{pe}(\z,\pd)}

\newcommand{\Vdr}{V_\text{dr}(\p)\xspace}
\newcommand{\Vpe}{V_\text{pe}(\p)\xspace}
\newcommand{\Vc}{V_\text{c}(\p)\xspace}
\newcommand{\cor}{f_{\varepsilon,x}(\tau)}
\newcommand{\ecor}{\hat{f}_{\varepsilon,x}(\tau)}

\newcommand{\eCor}{\hat{\bm f}_{\varepsilon,x}}

\newcommand{\ee}{\varepsilon(t,\p)\xspace}
\newcommand{\eek}{\varepsilon_{K}(t,\p)\xspace}
\newcommand{\elin}{\varepsilon_\text{lin}(t,\p)\xspace}
\newcommand{\elind}{\varepsilon_\text{lin}(t,\pd)\xspace}
\newcommand{\enl}{\varepsilon_\text{nl}(t,\p)\xspace}
\newcommand{\enld}{\varepsilon_\text{nl}(t,\pd)\xspace}
\newcommand{\osapred}{\left(t \,\middle|\, t-1,\p\right)\xspace}
\newcommand{\uu}{\hat{u}\osapred}
\newcommand{\ulin}{\hat{u}_\text{lin}\osapred}
\newcommand{\unl}{\hat{u}_\text{nl}\osapred}
\newcommand{\reg}{\bm\varphi(t)}
\newcommand{\regt}{\bm\varphi\tran(t)}
\newcommand{\regk}{\bm\varphi_K(t)}
\newcommand{\regkt}{\bm\varphi_K\tran(t)}
\newcommand{\uk}{\utv_K(t)\xspace}
\newcommand{\zzz}{\bm\zeta(t)\xspace}
\newcommand{\zzt}{\bm\zeta\tran(t)\xspace}

\makeatletter
\newcommand{\pushright}[1]{\ifmeasuring@#1\else\omit\hfill$\displaystyle#1$\fi\ignorespaces}
\newcommand{\pushleft}[1]{\ifmeasuring@#1\else\omit$\displaystyle#1$\hfill\fi\ignorespaces}
\makeatother

\articletype{Article Type}%
\received{}
\revised{}
\accepted{}

\begin{document}
\title{Data-driven load disturbance rejection}
\author[1]{Róger W. P. da Silva}
\author[2]{Diego Eckhard}
\authormark{DA SILVA and ECKHARD}
\address[1]{\orgdiv{Programa de Pós-Graduação em Engenharia Elétrica}, \orgname{Universidade Federal do Rio Grande do Sul},
\orgaddress{\state{RS}, \country{Brazil}}}
\address[2]{\orgdiv{Dept. Mat. Pura e Aplicada}, \orgname{Universidade Federal do Rio Grande do Sul},
\orgaddress{\state{RS}, \country{Brazil}}}
\corres{Diego Eckhard, Dept. Mat. Pura e Aplicada, Universidade Federal do Rio Grande do Sul, Av. Bento Gonçalves, 9500, Prédio 43111, Sala B218, Agronomia
91509900 - Porto Alegre, RS - Brazil.\\
\email{diegoeck@ufrgs.br}}
%\presentaddress{<Present address>}
\abstract[Summary]{Data-driven direct methods are still growing in popularity almost three decades after they were introduced.
These methods use data collected from the process to identify optimal controller's parameters with little knowledge about the process itself.
However, most of those works focus on the problem of reference tracking, whereas many of the problems faced in real-life are of disturbance rejection or attenuation.
Also, the vastly majority of those works identify the parameters of linearly parametrized controllers, which amounts to fixing the poles of the controller's transfer function.
Although the identification of the controller's poles is not prohibitive, as hinted by some of the papers, there is little effort on presenting a data-driven solution capable of doing so.
With all that in mind, this work proposes a data-driven approach which is able to identify the zeros and the poles of a linear controller aiming at disturbance rejection.
Two different one-step ahead predictors are proposed, one that is linear on the parameters and another that is non-linear.
Also, two different techniques are employed to estimate the controller parameters, the first one minimizes the quadratic norm of the prediction error while the second one minimizes the correlation between the prediction error and an external signal.
Simulations show the effectiveness of the proposed methods to estimate the optimal controller parameters of restricted order controllers aiming at disturbance rejection.}
\jnlcitation{\cname{%
\author{da Silva, Roger W. P.}, and
\author{Eckhard, Diego} (\cyear{<year>}),
\ctitle{Data-driven load disturbance rejection with the correlation approach}, \cjournal{Int J Adapt Control Signal Process} <year> <vol> Page <xxx>-<xxx>}
%\footnotetext{\textbf{<abbreviation head:>} <abbreviations> ..
}
\maketitle
\section{Introduction}

Direct data-driven methods for controller design form a category of methods that serve for adjusting a controller's parameters without the needing of a process model.
These methods rely on the availability of some input-output data collected from the process \cite{bazanella2011data}.
Some methods such as the iterative feedback tuning (IFT) \cite{hjalmarsson1998Aiterative}, require that those data are collected during a specific experiment with a uniquely crafted excitation signal, while others work with any signal, as long as it is rich enough.
Some classical methods require performing a sequence of experiments while changing the controller after each iteration, such as the IFT \cite{hjalmarsson1998Aiterative}, the correlation-based tuning (CBT)\cite{karimi2004iterative}, and the frequency-domain tuning (FDT) \cite{kammer2000direct}.
However, for some methods a single experiment is sufficient as it is the case of the non-iterative correlation-base tuning (NCbT) \cite{karimi2007noniterative} or the optimal controller identification (OCI) \cite{campestrini2017data}.
Methods that need a sequence of experiments forcing the characteristics of the excitation signal are called \textbf{iterative}, while the others are called \textbf{one-shot} or simply \textbf{non-iterative}.
Some methods also use an instrumental variable that must be generated from noise-uncorrelated data, to which a repeated experiment is the most common solution, as is the case for the virtual reference feedback tuning (VRFT),\cite{campi2002virtual} and others that followed.
% TODO: colocar os outros

In general, data-driven methods search for the optimal controller parameters by optimizing a cost function that is calculated from the available data.
Many methods aim at \textbf{model matching} that is minimizing the norm of the difference between the closed loop transfer function and a prescribed model, called the \textbf{reference model}.
Examples of those include some classical works\cite{campi2002virtual,karimi2007noniterative}, some newer \cite{campestrini2017data}, and many works that followed.
% TODO: colocar estes trabalhos
Other methods, though, have different objectives, such as shaping the $H_\infty$ norm of sensitivity functions\cite{formentin2013data} or dealing with optimal control.\cite{dasilva2018data}
In general, the difference between the closed-loop transfer function and the reference model is non-linear with respect to the controller's parameters, which leads to non-convex optimizations that needs to be solved iteratively.
Some methods approximate that non-convex cost function by a convex one when the controller is \textbf{linearly parametrized}, which means the controller structure is represented by a linear combination of more basic transfer functions.

Although the methods perform fairly well when the controller structure is of full order, in comparison with model-based approaches, they really stand out when tuning a restricted-order controller, as already shown in the literature\cite{campestrini2017data,formentin2014comparison}.
%This is probably for the best, since controllers with PID structure are among the most commonly found in industrial applications.
%That is arguably because they are very versatile controllers: the integral term allows for good regulation or tracking of stepwise reference signals, while also rejecting stepwise disturbances; the proportional and derivative terms %lead to some additional performance gain in many cases, as indicated in any good control book.\cite{aastrom1995pid}
%Additionally, if the derivative term's pole is made free the controller may benefit from less noise amplification and increased performance, since it becomes closer to a lead-lag compensator.
%These controllers are suitable for many control tasks such as reference tracking, regulation, and disturbance rejection.
%
Some control literature indicates\cite{szita1996model} that most of industrial applications are more concerned with disturbance rejection, but most of the model-matching techniques focus on set-point response.
Surprisingly, most of the works within the data-driven literature still aim only at reference tracking, shaping the closed-loop complementary sensitivity function.
This work is an attempt to increase the amount of research in the direction of data-driven disturbance rejection with model matching.
The very meaning of disturbance rejection may be too broad\cite{gao2014centrality} and this work deals primarily with load disturbance rejection, addressing the same problem as Reference \citenum{eckhard2018virtual}, although with different approaches.
The ideas presented in this article extend that work and the work of Reference \citenum{dasilva2019extension}.
The work in Reference \citenum{eckhard2018virtual} seeks to shape the load disturbance sensitivity function response to a given disturbance, using virtual input-output signals to identify the optimal controller in a way very similar to the VRFT method. It is developed only for linearly parametrized controllers minimizing a quadratic norm of an error function.
On the other hand, Reference \citenum{dasilva2019extension} presents a modified version of the NCbT method \cite{karimi2007noniterative} which proposes to shape the load disturbance sensitivity function instead of the complementary sensitivity.

This article proposes to generate virtual signals as the ones proposed in Reference \citenum{eckhard2018virtual} and use them to identify the optimal controller. 
Here, the identification is performed either through the minimization of the prediction error's norm,\cite{eckhard2018virtual} or by the minimization of the correlation between an error variable and the experiment input \cite{dasilva2019extension}.
The main differences that sets this work apart from those other two are as follows.
Unlike the works in Reference \citenum{eckhard2018virtual}  and \citenum{dasilva2019extension}, this article is not restricted to linearly parametrized controllers such that the controller structure proposed here allows for the identification of the controller's poles, using two different one-step ahead predictors. Also, we propose the use of two different optimization algorithms, the first one that minimizes the norm of the one-step ahead predicition error and the second that minimizes the correlation of the one-step ahead error and an external signal.

The rest of this article is organized as follows: Section~\ref{sec:model-matching-design} establishes the problem solved along the article, and Section~\ref{sec:data-driven-errors} presents how the prediction errors are obtained from data.
Then, Section~\ref{sec:prediction-error-norm} shows how to estimate the controller's parameters by minimizing the prediction error's norm, while Section~\ref{sec:correlation-error-input} shows how to do the same by minimizing the correlation between the prediction error and the experiment input.
Section~\ref{sec:examples} gives some illustrative examples and Section~\ref{sec:conclusion} draws some conclusions and presents some open matters.

\section{Model Matching Design}
\label{sec:model-matching-design}
\subsection{Preliminaries}

Consider a linear time-invariant discrete-time single-input single-output process
\begin{equation}\label{eq:process}
\Yt=\Gq \Ut + \Vt
\end{equation}
where $\Ut$ is the process input, $\Vt$ is the output noise, which is a stochastic process with zero mean, $\Gq$ is the process transfer operator and $\z$ is the time shift operator $\z x(t)=x(t+1)$.
The process input signal is composed by two terms:
\begin{align}
\Ut &= \Cqp \left[ \Rt - \Ytp \right] + \Dt = \Cqp \Etp + \Dt,\label{eq:process-input}
\end{align}
the first one is the control action that is calculated from the error $\Etp$ between the reference input $\Rt$ and the measured closed-loop output $\Ytp$, the second term is the external disturbance $\Dt$.
Finally, $\Cqp$ is a linear time-invariant controller that is parametrized by $\p \in \mathbb{R}^n$, and belongs to the following user-defined controller class
\begin{equation}
  \mathcal C = \left\{ \Cqp, \p \in \mathcal{P} \subseteq \mathbb{R}^n\right\}.\label{eq:class}
\end{equation}
The block diagram of the final closed-loop system is presented in Figure~\ref{fig:sys-cl}.

\begin{SCfigure}
  \includegraphics{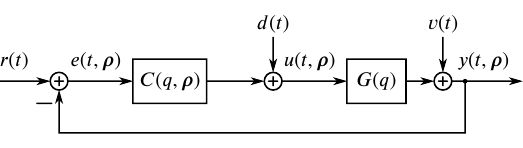}
  \caption{\label{fig:sys-cl}Closed-loop block diagram.}
\end{SCfigure}

The system \eqref{eq:process}--\eqref{eq:process-input} in closed-loop reads:
\begin{align}
\Ytp &= \Tqp\Rt + \Qqp\Dt + \Sqp\Vt \label{eq:y-cl}\\
\Utp &= \Gqi\Tqp\Rt + \Sqp\Dt - \Cqp\Sqp\Vt \label{eq:u-cl}
\end{align}
where the following sensitivity functions are employed:
\begin{align}
  \Tqp &= 1 - \Sqp \label{eq:tqp}\\
  \Sqp &= \left[ 1 + \Gq\Cqp \right]\inv \label{eq:sqp}\\
  \Qqp &= \Gq\Sqp\label{eq:qqp}.
\end{align}
Observe, from \eqref{eq:sqp} and \eqref{eq:qqp} that, for any given controller,
\begin{equation}
  \Qiqp = \Gqi + \Cqp.\label{eq:q-inv}
\end{equation}

% NOTE: Disturbance Model já existe e é como Noise Model, serve pra modelar o sinal exógeno da perturbação. Tem que escolher um outro nome. Sugestões: Disturbance Response design, Input Sensitivity design, o velho Reference Model design ou Model Matching Design for disturbance response etc.
The role of the control designer is to choose the parameter vector $\p$ of the controller in order to obtain good performance for the closed-loop system.
Disturbance response design requires the choice of a \textit{desired output} for the closed-loop system considering a specified disturbance signal, that is,
\begin{equation}
  \Yd = \Qdq \Dt,
\end{equation}
where $\Qdq$ is the \textit{reference model}, i.e. the desired load disturbance sensitivity.
The paramaters of the controller are obtained as the solution of the optimization problem:
\begin{equation}
  \pstar = \arg \min_{\p} ~ \Vdr,\label{eq:rho-opt}
\end{equation}
with the disturbance response criterion defined as
\begin{equation}\label{eq:v-dr}
  \Vdr \triangleq \bExp \left[ \left(\Qdq - \Qqp\right) \Dt\right]^2,
\end{equation}
where
\begin{equation}
  \bExp f(t) \triangleq \lim_{N \rightarrow \infty} \frac{1}{N} \sumN \Exp f(t),
\end{equation}
as usually defined,\cite{ljung1999system} and with $\Exp \cdot$ denoting expectation.
An usual choice for the reference model\footnote{considering minimum phase systems} is to include a zero at $1$ to ensure null steady-state gain for constant disturbances.
% and to choose the poles accordingly to the desired settling time.
% TODO: citar Virgínia e referências

The \textit{ideal controller} is the one that makes the closed-loop system to match the reference model:
\begin{equation}
  \Cdq = \Qdiq - \Gqi \label{eq:c-d}.
\end{equation}
Observe that the knowledge of the process model $\Gq$ is needed to compute this controller.
Also, this ideal controller \eqref{eq:c-d} may not be in the controllers class \eqref{eq:class}. When it does, the following assumption holds.

\begin{assumption}[Controller matching]\label{ass:matching}
The ideal controller in \eqref{eq:c-d} may be represented with the available controller structure,
\begin{equation}
  \exists~\pd\in \mathcal{P} \text{ such that } C(\z,\pd) = \Cdq,
\end{equation}
or, equivalently,
\begin{equation}
  \Cdq \in \mathcal{C}.
\end{equation}
\end{assumption}

When Assumption~\ref{ass:matching} does not hold, the obtained controller $C(\z,\pstar)$ is not the ideal controller, but it is the best controller that can be used, i.e., the optimal.
In other words, it minimizes \eqref{eq:v-dr} resulting in a closed-loop response that is as close as possible to the desired output $\Yd$.
The next section presents techniques to estimate the optimal controller from experimental data, without the need of a process model.

\section{Data-driven prediction errors}
\label{sec:data-driven-errors}

Previous section presented the disturbance response design, where the user chooses a reference model $\Qdq$ and the parameters of the controller are computed solving an optimization problem.
The objective function depends on $\Qqp$, which, in turn, depends on the process transfer function $\Gq$.
Data-driven design techniques assume that the process transfer function is unknown, but the user can collect a sufficiently rich batch of process input and output data which contains implicitly information about the process.
Suppose the data contain $N$ samples coming from an open-loop experiment in the process, then we define the following dataset:
\begin{equation}\label{eq:z-ol}
  \mathcal{Z}_\text{ol}^N = \lbrace u(1),\,y(1),\,u(2),\,y(2),\,\dots,\,u(N),\,y(N)\rbrace.
\end{equation}
On the other hand, if the data come from a closed-loop experiment where the reference input is excited, then we define
\begin{equation}\label{eq:z-cl}
  \mathcal{Z}_\text{cl}^N = \lbrace r(1),\,u(1),\,y(1),\,r(2),\,u(2),\,y(2),\,\dots,\,r(N),\,u(N),\,y(N)\rbrace
\end{equation}
as the dataset to be employed.

\begin{assumption}[Noise and input uncorrelated\label{ass:noise}]
Both $\Ut$ (in open-loop) and $\Rt$ (in closed-loop) are assumed to be quasi-stationary and uncorrelated with the noise, that is,
\begin{align}
  &\bExp \Ut v(s) = 0 ;\forall t,s, \\
  &\bExp \Rt v(s) = 0 ;\forall t,s
\end{align}
\end{assumption}

In the sequence we present a framework that uses one of the above datasets, to estimate the parameters $\p$ of the controller without using a model of the process.
% TODO: decidir onde é \pd e onde é \pstar
The objective of this section is to introduce the prediction error employed by that framework.

To identify the controller, the following thought exercise is proposed in Reference \citenum{eckhard2018virtual}:
we know for a fact that the data is collected during an open-loop or closed-loop in the real process, but pretend, for a moment, that the real dataset comes from a closed-loop experiment with the ideal controller in the loop.
This imagined experiment is presented in Figure~\ref{fig:virtual-sys}, the real-life experiment is drawn in black, the virtual-only experiment is in blue.
\begin{SCfigure}
  \includegraphics{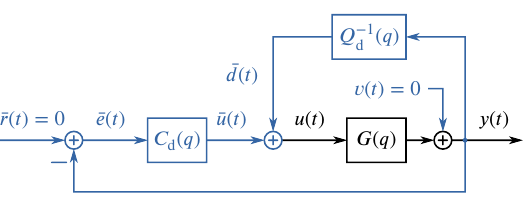}
  \caption{\label{fig:virtual-sys} Virtual closed-loop block diagram showing the real (black) and the virtual (light blue) signals.}
\end{SCfigure}
Now, pretend that the virtual experiment is a noise-free, regulation setup.
Because this is a regulation setup, the ideal controller's input, i.e. the \emph{virtual error}, is simply given by
\begin{equation}
  \Etv = - \Yt.\label{eq:e-virt}
\end{equation}
Since both the noise $\Vt$ and the \emph{virtual reference} $\Rtv$ are null, the measured signal $\Yt$ is only due to the effect of a \emph{virtual disturbance} $\Dtv$ acting in the process input.
Therefore, from \eqref{eq:y-cl}, the virtual disturbance is given by
\begin{equation}
  \Dtv = \Qdiq \Yt,
\end{equation}
whereas the ideal controller's output is the \emph{virtual control action}, obtained from the first term of \eqref{eq:process-input} and given by
\begin{align}
  \Utv &= \Ut - \Dtv = \Ut - \Qdiq \Yt.\label{eq:u-virt}
\end{align}
This concludes the thought exercise, giving the input/output signals \eqref{eq:e-virt} and \eqref{eq:u-virt} needed to identify the controller.
In other words:
Even though the plant $\Gq$ is unknown, when it is fed by \Ut, it generates \Yt as output, so, a ``good'' controller is one that generates $\Utv$ when fed by $\Etv$.
Because those virtual signals are generated from data, the controller design can be seen as the data-driven identification of the dynamical relation between them.

Notice that part of the controller may be know or is fixed, such that it should not be identified, for instance if an integral action is required in the controller.
Therefore, we employ the following more general controller structure along the remaining of the paper:
Let us then employ the following more general controller structure:
\begin{equation}
  \Cqp = \Ciqp\Cfq,\label{eq:ci-cf}
\end{equation}
where $\Ciqp$ is the unknown portion of the controller, which must be \emph{identified}, while $\Cfq$ is some known \emph{fixed} portion.
That fixed portion may include the integral action for instance, and the whole structure is depicted in Figure~\ref{fig:cf-ci}.
\begin{SCfigure}
  \centering
  \includegraphics{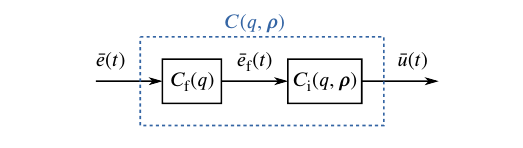}
  \caption{\label{fig:cf-ci}Closed-loop block diagram and the virtual system's signals.}
\end{SCfigure}
We propose in this article to identify the transfer operator $\Ciqp$ using the following input/output relationship
\begin{equation}
  \Utv = \Ciqp \Efv, \label{eq:uf-virt}
\end{equation}
where the new intermediary error is given by
\begin{align}
  \Efv &= \Cfq \Etv = -\Cfq \Yt. \label{eq:ef-virt}
\end{align}
Also, it is assumed that $\Ciqp$ depends on the parameters accordingly with the following classical parametrization:
\begin{align}
  \Ciqp &= \frac{\Bqp}{\Aqp} = \frac{b_0 + b_1\z^{-1} + \cdots + b_{n_b} \z^{-n_b}}{1 + a_1\z^{-1} + \cdots + a_{n_a} \z^{-n_a}},\label{eq:c-i}
\end{align}
whereas the parameters vector is defined as
\begin{equation}
  \p\tran = \begin{bmatrix}
    a_1 & a_2 & \dots & a_{n_b} & b_0 & b_1 & \dots & b_{n_b}
  \end{bmatrix}.
\end{equation}

From the systems identification theory, an one-step ahead predictor $\uu$ of the controller's output may be constructed, which is used to compute the following \emph{prediction error} signal:
\begin{equation}
  \ee = \Utv - \uu.\label{eq:epred}
\end{equation}

Many system identification algorithms may be employed that minimize some function of the error signal \eqref{eq:epred} in order to estimate the controller's parameters.
In this work, we propose the use of two different approaches, one that minimizes the 2-norm of the error and another that minimizes the correlation of the error with an instrumental variable.
These two approaches are described in the next sections.
However, before moving to the next section, let us propose two different predictors for the controller output.

\subsection{Linear predictor}

The first one-step ahead predictor for the controller output is inspired on the ARX (AutoRegressive with eXogenous term) predictor and is given by
\begin{equation}
  \ulin = \left[ 1 - \Aqp \right]\Utv + \Bqp \Efv,\label{eq:u-lin}
\end{equation}
which may be written in a linear regression form as
\begin{equation}
  \ulin = \regt\p,
\end{equation}
where the regressor vector is defined as
\begin{align}
  \regt = \begin{bmatrix}
    -\utv(t-1) &
    \cdots &
    -\utv(t-n_a)&
    \efv(t)&
    \cdots&
    \efv(t-n_b)
  \end{bmatrix}.\nonumber\\
\end{align}
Therefore, the prediction error for the linear predictor is
\begin{align}
  \elin &= \Utv - \ulin \label{eq:epred-lin-1}\\
    &= \Utv - \regt \p.\label{eq:epred-lin}
\end{align}
This predictor is linear on $\p$ which simplifies the analysis and design of identification algorithms.

\subsection{Nonlinear predictor}

The second one-step ahead predictor for the controller output is inspired on the OE (Output Error) predictor and, as such, is given by
\begin{align}
  \unl &= \Ciqp \Efv = \frac{\Bqp}{\Aqp} \Efv, \label{eq:u-nlin}
\end{align}
which makes the prediction error simply
\begin{equation}
  \enl = \Utv - \Ciqp \Efv. \label{eq:epred-nlin}
\end{equation}
Observe that this predictor is non-linear w.r.t. the parameters $\p$, because some parameters appear on its denominator.
Notice that the two proposed predictors are equivalent if $n_a=0$, i.e. if all parameters to be identified are on the numerator of the controller or, equivalently, $\Aqp = 1$.
In the next sections, we will propose the use of two different approaches to identify the parameters of the controller using both the linear \eqref{eq:u-lin} and the non-linear \eqref{eq:u-nlin} predictors.

\begin{remark}
  Observe that when the whole controller's denominator is fixed, i.e. $n_a = 0$ and $\Aqp = 1$, the entire controller's transfer operator becomes \emph{linearly parametrized}.
  Solutions for identifying the controller with this particular constraint have already been proposed before.\cite{eckhard2018virtual,dasilva2019extension}.
\end{remark}

\section{Minimize the 2-norm of the error}
\label{sec:prediction-error-norm}

One possibility for identifying the controller's parameter is minimizing the 2-norm of the prediction error, leading to the well-known least squares and prediction error methods.
This idea has already been employed with data-driven methods such as the VRFT\cite{campi2002virtual} and the VDFT.\cite{eckhard2018virtual}

Following this approach, the controller parameters are estimated by solving and optimization problem that minimizes the 2-norm of a filtered version of the error signal, that is,
\begin{equation}
\hat{\p} = \arg \min_{\p} \Vpe,
\end{equation}
with the prediction error cost function defined as
\begin{equation}
  \Vpe = \frac{1}{N} \sumN \abs{\eek}^2,\label{eq:v-pe}
\end{equation}
where $\eek = \Kqp \ee$ is the prediction error filtered through the filter $\Kqp$.
The latter may be seen as an extra degree of freedom used to modify the objective function.
Some choices for this filter will be given for each method later.

Now, the user may choose to use the linear or the non-linear predictors, resulting in two different algorithms to estimate the controller parameters.

\subsection{Linear predictor}

Observe that by replacing \eqref{eq:u-lin} in \eqref{eq:epred-lin-1}, the linear prediction error may be written
\begin{equation}
  \elin = \Aqp \Utv - \Bqp \Efv.\label{eq:epred-lin-2}
\end{equation}
Also, note that applying the definitions \eqref{eq:ci-cf} and \eqref{eq:c-i} in \eqref{eq:q-inv} gives
\begin{equation}
  \Aqp\Qiqp = \Aqp\Gqi + \Bqp\Cfq. \label{eq:aq-inv}
\end{equation}
Using these two equalities, we will calculate the noiseless prediction error for the linear predictor, and then propose a solution for the parameters estimate.

% ***

\begin{lemma}\label{lem:epred-lin-noiseless}
  In the absence of noise, the prediction error obtained from the linear predictor is given by
  \begin{equation}
    \elin \Big\lvert_{\Vt=0} = \Aqp\left[ \Qiqp - \Qdiq \right]\Yt,\label{eq:epred-lin-noiseless}
  \end{equation}
  regardless if the data are collected in open-loop or closed-loop.
\end{lemma}
\begin{proof}
  First consider the case where the linear prediction error \eqref{eq:epred-lin-2} is obtained from open-loop data. In this case, the excitation signal is $\Ut$ and the prediction error is given by
  % TODO: talvez mover um pouco disso pra cima --> *** ?
  \begin{align}
    \elin &= \Aqp \left[ \Ut - \Qdiq \Yt \right] + \Bqp \Cfq \Yt \label{eq:epred-lin-3}\\
      &= \left[\Aqp \Gqi + \Bqp\Cfq \right] \Gq \Ut - \Aqp \Qdiq \Gq\Ut - \left[ \Aqp\Qdiq - \Bqp\Cfq \right]\Vt \label{eq:epred-lin-4}\\
      &= \Aqp \left[ \Qiqp - \Qdiq \right]\Gq\Ut - \left[ \Aqp\Qdiq - \Bqp\Cfq \right]\Vt, \label{eq:epred-lin-ol}
  \end{align}
  where \eqref{eq:epred-lin-3} comes from \eqref{eq:u-virt}, \eqref{eq:ef-virt} and \eqref{eq:epred-lin-2}, while \eqref{eq:epred-lin-4} comes from the open-loop relation \eqref{eq:process},
  and \eqref{eq:epred-lin-ol} comes from \eqref{eq:aq-inv}.
  Removing the noise-dependant term from \eqref{eq:epred-lin-ol} yields \eqref{eq:epred-lin-noiseless}.

  Now, consider the case where the data come from a closed-loop experiment with an initial controller $\Coq$.
  In this case, the excitation signal is $\Rt$, and the process input/output signals from \eqref{eq:u-cl} and \eqref{eq:y-cl} are given by
  \begin{align}
    \Ut &= \Gqi \Toq \Rt - \Coq \Soq \Vt,\label{eq:u0-cl}\\
    \Yt &= \Toq\Rt + \Soq\Vt,\label{eq:y0-cl}
  \end{align}
  where $S_0(q) = \left[ 1 + G(q)C_0(q)\right]\inv$ and $T_0(q) = 1 - S_0(q)$.
  Replacing \eqref{eq:u0-cl} and \eqref{eq:y0-cl} in \eqref{eq:epred-lin-3} and then using \eqref{eq:aq-inv} again allows us to write the closed-loop prediction error as
  \begin{align}
    \elin &= \left[ \Aqp\Gqi + \Bqp\Cfq \right] \Toq\Rt
        - \Aqp\Qdiq \Toq\Rt \nonumber\\
        &\quad {} - \left[ \Aqp\Qdiq - \Bqp\Cfq \right]\Soq\Vt
        - \Aqp\Coq \Soq\Vt \nonumber \\
      &= \Aqp \left[ \Qiqp - \Qdiq \right] \Toq\Rt \nonumber\\
        &\quad {} - \left[ \Aqp\Qdiq - \Bqp\Cfq \right]\Soq\Vt
        - \Aqp\Coq \Soq\Vt. \label{eq:epred-lin-cl}
  \end{align}
  Removing the noise-dependant term from \eqref{eq:epred-lin-cl} yields \eqref{eq:epred-lin-noiseless}.
\end{proof}

\begin{theorem}
  When Assumption~\ref{ass:matching} holds and the linear predictor \eqref{eq:u-lin} is calculated from noiseless data, the ideal parameters map to the global minimum of
  \begin{equation}
    \Vpe = \frac{1}{N} \sumN \abs*{\Kqp \left[ \Utv - \regt\p \right]}^2.\label{eq:v-pe-lin}
  \end{equation}
\end{theorem}
\begin{proof}
  When Assumption~\ref{ass:matching} holds, $Q(\z,\pd) = \Qdq$, then by Lemma~\ref{lem:epred-lin-noiseless} the error evaluates to zero, and so does the cost:
  \begin{equation}
    \Vped \Big\lvert_{\Vt=0} = 0,
  \end{equation}
  which is the global minimum because the cost \eqref{eq:v-pe-lin} is quadratic.
\end{proof}

\begin{corollary}
  When Assumption~\ref{ass:matching} holds and when using the linear predictor \eqref{eq:u-lin}, the ideal parameters may be estimated as
  \begin{equation}
    \hat\p = \left[ \sumN \regk \regkt \right]\inv \sumN \regkt \uk,\label{eq:rho-pe-lin}
  \end{equation}
  where $\regk = \Kqp\reg$ and $\uk = \Kqp \Utv$.
\end{corollary}

The greatest advantage of using the linear predictor is that the solution of the optimization problem \eqref{eq:v-pe} is analytic and given by \eqref{eq:rho-pe-lin}.
If the level of noise is low, the bias and variance of the estimate are small and this method is effective to identify the optimal controller.
However, when the data are collected with noise $v(t) \neq 0$, this is a classical \textit{errors-in-variables} identification problem, for which this particular solution is known to produce biased estimates when the noise is not negligible.
In that case, other solutions may be explored.\cite{soderstrom2018errors}

\subsubsection{Filter Design}

When Assumption~\ref{ass:matching} holds, the solution of \eqref{eq:rho-pe-lin} gives a good estimate for the ideal parameters, otherwise, the estimate may be far from the optimal parameters.
This may be remediated by designing the filter $\Kqp$ with that goal in mind.

\begin{theorem}
  When using noiseless data and the linear predictor \eqref{eq:u-lin}, the filter that will force the minimum of the prediction error cost \eqref{eq:v-pe-lin} to be the same that the minimum of the disturbance response cost \eqref{eq:v-dr} is
  \begin{equation}
    \abs*{\Kwp}^2 = \frac{\abs*{\Qwp\Qdw}^2}{\abs*{\Awp}^2}\frac{\Phid}{\Phiy},\label{eq:filter-pe-lin-1}
  \end{equation}
  where $\Phid$ and $\Phiy$ are the power spectra of the disturbance to be rejected $\Dt$ and of the output collected $\Yt$, respectively.
\end{theorem}
\begin{proof}
  In the noiseless case, using Parseval's theorem and Lemma~\ref{lem:epred-lin-noiseless}, we are able to write the cost \eqref{eq:v-pe-lin} as
  \begin{align}
    \Vpe &= \frac{1}{2\pi}\int_{-\pi}^{\pi}\abs*{\Kwp\vphantom{\Big|}}^2 \abs*{\Awp\vphantom{\Big|}}^2
        \abs*{ \Qiwp - \Qdiw }^2 \Phiy\dd\omega.\label{eq:v-pe-parseval}
  \end{align}
  Compare it with the frequency-domain version of the disturbance response cost \eqref{eq:v-dr}
  \begin{align}
    \Vdr &= \frac{1}{2\pi}\int_{-\pi}^{\pi}\abs*{ \Qdw - \Qwp }^2 \Phid\dd\omega,\label{eq:v-dr-parseval}
  \end{align}
  also obtained with Parseval's theorem, and where $\Phid$ is the power spectrum of the disturbance to be rejected $\Dt$.
  Equating the right side of \eqref{eq:v-pe-parseval} and \eqref{eq:v-dr-parseval} and solving for the filter, then using the identity
  \begin{equation}
    \frac{\Qdw - \Qwp}{\Qiqp - \Qdiw} = \Qwp\Qdw,
  \end{equation}
  results in \eqref{eq:filter-pe-lin-1}, which is the filter that makes the two functions the same.
\end{proof}

\begin{remark}
  Observe that the filter \eqref{eq:filter-pe-lin-1} may be approximated by
  \begin{equation}
  \abs*{\Kwp}^2 \approx \frac{\abs*{\Qdw\Qdw}^2}{\abs*{\Awp}^2}\frac{\Phid}{\Phiy},\label{eq:filter-pe-lin}
  \end{equation}
  as long as $\Qqps$ and $\Qdq$ are similar.
\end{remark}

\begin{remark}
  Note that all the quantities involved in \eqref{eq:filter-pe-lin} are available or may be estimated from data.
  Unfortunately, even after the approximation, the filter still depends on the parameters; therefore, an iterative optimization procedure must be employed.
  A possible option is to start with $A(q) = 1$ in the filter's denominator, then solve a least squares problem and update the filter at each iteration.
\end{remark}

\begin{remark}
  The filter proposed is optimal in the absence of noise.
  However, if there is significant noise, the filter may amplify the noise effects and the minima will not be close.
  The same is valid for the other filters proposed later.
\end{remark}

\subsection{Nonlinear predictor}

As before, let us calculate now the prediction with the non-linear predictor and noiseless data.

% ***

\begin{lemma}\label{lem:epred-nlin-noiseless}
  In the absence of noise, the prediction error obtained from the non-linear predictor \eqref{eq:u-nlin} is given by
  \begin{equation}
    \ee = \left[ \Qiqp - \Qdiq \right] \Yt,\label{eq:epred-nlin-noiseless}
  \end{equation}
  regardless if the data are collected open-loop or closed-loop.
\end{lemma}
\begin{proof}
  First consider open-loop data.
  In this case, the excitation signal is $\Ut$ and the prediction error from \eqref{eq:epred-nlin} is
  % TODO: talvez mover um pouco disso pra antes --> *** ?
  \begin{align}
    \enl &= \Utv - \Cqp \Etv \label{eq:epred-nlin-1}\\
      &= \Ut - \Qdiq \Yt + \Cqp \Yt \label{eq:epred-nlin-2}\\
      &= \left[ \Gqi + \Cqp - \Qdiq \right] \Gq \Ut
        - \left[ \Qdiq - \Cqp \right]\Vt \label{eq:epred-nlin-3}\\
      &= \left[\Qiqp - \Qdiq \right] \Gq \Ut
        - \left[ \Qdiq - \Cqp \right]\Vt, \label{eq:epred-nlin-ol}
  \end{align}
  where \eqref{eq:epred-nlin-1} comes from \eqref{eq:ci-cf}, \eqref{eq:ef-virt}, and \eqref{eq:epred-nlin}.
  \eqref{eq:epred-nlin-2} comes from \eqref{eq:u-virt}, while \eqref{eq:epred-nlin-3} comes from the open-loop relation \eqref{eq:process},
  and \eqref{eq:epred-nlin-ol} comes from \eqref{eq:q-inv}.
  Removing the noise-dependant term from \eqref{eq:epred-nlin-ol} results in \eqref{eq:epred-nlin-noiseless}.

  Now, consider the case where the data are collected during a closed-loop experiment with an initial controller $C_0(q)$.
  In this case the excitation signal is $\Rt$ and, again, replacing \eqref{eq:u0-cl} and \eqref{eq:y0-cl} in \eqref{eq:epred-nlin-2} allows us to write the prediction error as
  \begin{align}
    \enl &= \left[\Gqi + \Cqp - \Qdiq \right] \Toq \Rt
        - \left[ \Qdiq - \Cqp + \Coq \right]\Soq \Vt \nonumber\\
      &= \left[\Qiqp - \Qdiq \right] \Toq \Rt
        - \left[ \Qdiq - \Cqp + \Coq \right]\Soq \Vt, \label{eq:epred-nlin-cl}
  \end{align}
  where \eqref{eq:epred-nlin-cl} comes from \eqref{eq:q-inv}.
  After removing the noise-dependant term from \eqref{eq:epred-nlin-cl} the prediction error becomes the same as the one in \eqref{eq:epred-nlin-noiseless}.
\end{proof}

\begin{theorem}
  When Assumption~\ref{ass:matching} holds and the non-linear predictor \eqref{eq:u-nlin} is calculated from noiseless data, the ideal parameters map to the global minimum of
  \begin{equation}
    \Vpe = \frac{1}{N} \sumN \abs*{\Kqp \left[ \Utv - \frac{\Bqp}{\Aqp}\Efv \right]}^2,\label{eq:v-pe-nlin}
  \end{equation}
\end{theorem}
\begin{proof}
  When Assumption~\ref{ass:matching} holds, $Q(\z,\pd) = \Qdq$, then from Lemma~\ref{lem:epred-nlin-noiseless} the error evaluates to zero, and so does the cost.
  This is a global minimum because the cost \eqref{eq:v-pe-nlin} is non-negative.
\end{proof}

The cost function \eqref{eq:v-pe-nlin} is clearly non-convex.
Therefore, the solution of the optimization problem must be obtained numerically.
There are many methods that are able to find local minima of non-convex functions such as the steepest descent, the Newton method and the Nelder-Mead algorithm.
If those methods are given an initial estimate close to the global minimum they are able to converge to the correct parameters.
When using Matlab, the Nelder-Mead algorithm is available by default as the \texttt{fminsearch} command, whereas \texttt{lsqnonlin} is another option if the Optimization Toolbox is installed.

In case there is no noise and Assumption~\ref{ass:matching} holds, the minimum of \eqref{eq:v-pe-nlin} is obtained at the ideal parameters.
However, if Assumption~\ref{ass:matching} does not hold, the minimum of \eqref{eq:v-pe-nlin} differs from that of \eqref{eq:v-dr}.
In this case, the filter may be employed, as before, to bring the two minima close to each other.

\subsubsection{Filter Design}

Now we show how the filter may be designed to compensate for the violation of Assumption~\ref{ass:matching} in the case of the non-linear predictor \eqref{eq:u-nlin}.

\begin{theorem}
  When Assumption~\ref{ass:matching} does not hold and the non-linear predictor \eqref{eq:u-nlin} is obtained from noiseless data, the filter that will force the minimum of the prediction error cost \eqref{eq:v-pe-nlin} to be the same as the minimum of the disturbance response cost \eqref{eq:v-dr} is
  \begin{equation}
    \abs*{\Kwp}^2 = \frac{\abs*{\Qwp\Qdw}^2\Phid}{\Phiy}.\label{eq:filter-pe-nlin-1}
  \end{equation}
\end{theorem}
\begin{proof}
  In the noiseless case, again using Parseval's theorem and Lemma~\ref{lem:epred-nlin-noiseless} the cost \eqref{eq:v-pe-nlin} may be written in the frequency domain as
  \begin{equation}
    \Vpe = \frac{1}{2\pi}\int_{-\pi}^{\pi}\abs*{\Kwp}^2 \abs*{ \Qiwp - \Qdiw }^2 \Phiy\dd\omega.
  \end{equation}
  A direct comparison with the disturbance response cost function in \eqref{eq:v-dr-parseval} shows that the filter \eqref{eq:filter-pe-nlin-1} turns the two functions the same.
\end{proof}

\begin{remark}
The filter \eqref{eq:filter-pe-nlin-1} may be approximated by
\begin{equation}
\abs*{\Kwp}^2 \approx \frac{\abs*{\Qdw\Qdw}^2\Phid}{\Phiy},\label{eq:filter-pe-nlin}
\end{equation}
as long as $\Qqps$ and $\Qdq$ are similar.
\end{remark}

\begin{remark}
  The quantities involved in \eqref{eq:filter-pe-nlin} are readily available or may be estimated from data.
  Also, note that in this case the approximate filter does not depend on the parameter.
  However, the iterative optimization is still needed because of the intrinsic non-convexity of the problem.
\end{remark}

\section{Minimize the correlation between input and error}
\label{sec:correlation-error-input}

As stated before, there is a second option for identifying the controller's parameter.
This option is minimizing the 2-norm of the cross-correlation between the one-step ahead prediction error and the experiment's excitation input.
To present this option, we first recall the following definition of the cross-correlation function between the prediction error $\varepsilon(t,\bm\rho)$ and another signal $x(t)$ driving the experiment:
\begin{equation}\label{eq:corr-inf}
  \cor = \lim_{L \to \infty} \frac{1}{2L+1}\sum_{t=-L}^{L}\eek x(t-\tau),
\end{equation}
where $\tau$ represents the lag and where the explicit dependency on the parameters is omitted.
Also $\eek = \Kqp\ee$ is the prediction error filtered through a filter $\Kqp$, as before, whose design is presented later.
A finite length approximation of the cross-correlation function is:
\begin{equation}\label{eq:corr-fin}
  \ecor = \frac{1}{N} \sumN \eek x(t-\tau).
\end{equation}
% This function also has finite length and its energy decreases as we move away from $\tau = 0$.
Consider now a vector formed by stacking $2L+1$ samples of \eqref{eq:corr-fin} for $-L \leqslant \tau \leqslant +L$:
\begin{equation}\label{eq:rex}
  \eCor = \frac{1}{N} \sumN \eek \zz,
\end{equation}
where $\zzz \in \mathbb{R}^{2L+1}$ has the samples of the excitation signal during the experiment and where $L$ is the maximum number of lags considered at each side of the cross-correlation function.
If the data are collected during an open-loop experiment, then the excitation signal is $\Ut$, which means
\begin{equation}
  \zzt = \begin{bmatrix} u(t+L) & \dots & \Ut & \dots & u(t-L) \end{bmatrix}.
\end{equation}
On the other hand, if the data come from a closed-loop experiment, then the excitation signal is $\Rt$, making
\begin{equation}
  \zzt = \begin{bmatrix} r(t+L) & \dots & \Rt & \dots & r(t-L) \end{bmatrix}.
\end{equation}

In any case, the estimation of the controller's parameters is done by minimizing the 2-norm of the approximate correlation between the experiment's excitation signal and the prediction error $\ee$.
The estimate is given by
\begin{equation}
\hat{\p} = \arg \min_{\p} \Vc,
\end{equation}
where
\begin{align}
  \Vc &= \frac{1}{2L+1} \sum_{\tau=-L}^{L} \abs*{\ecor}^2 = \frac{\norm*{\eCor}_2^2}{2L+1} = \frac{\norm*{\sumN \zzz \eek}_2^2}{N^2(2L+1)}. \label{eq:v-corr}
\end{align}

\subsection{Linear predictor}
Consider what happens when using the linear predictor \eqref{eq:u-lin} along with the correlation approach.
Let us start by considering noiseless data.

\begin{theorem}
  Under Assumptions~\ref{ass:matching} and~\ref{ass:noise}, and considering the affine prediction error \eqref{eq:epred-lin} obtained from noiseless data, the cost \eqref{eq:v-corr} has a global minimum at the ideal parameters.
\end{theorem}
\begin{proof}
  With noiseless data, from Lemma~\ref{lem:epred-lin-noiseless}, the prediction error vanishes at the ideal parameters, regardless if the it is obtained from open-loop or closed-loop data.
  Therefore, for noiseless data, the correlation between the input and the prediction error is also zero and $\Vcd = 0$ is a global minimum of \eqref{eq:v-corr}.
\end{proof}

\begin{corollary}
  When Assumptions~\ref{ass:matching} and~\ref{ass:noise} hold, the ideal parameters are estimated by
  \begin{equation}\label{eq:corr-lin-rho}
    \hat{\p} = \left[ \left( \sumN \zzz \regkt \right)\tran \left( \sumN \zzz \regkt \right)\right]\inv
      \left( \sumN \zzz \regkt \right)\tran  \left( \sumN \zzz \uk \right),
  \end{equation}
  which is the minimum of the quadratic function \eqref{eq:v-corr}.
\end{corollary}

Now, consider what happens when the experimental data are affected by some measurement noise.

\begin{lemma}\label{lem:epred-lin-noisy}
  When using noise-affected data, the prediction error calculated from the linear predictor \eqref{eq:u-lin} at the ideal parameters becomes
  \begin{equation}
    \elind = - \Aqpd \Gqi \Vt,\label{eq:epred-lin-noisy}
  \end{equation}
  regardless if the data are obtained from an open-loop or closed-loop experiment.
\end{lemma}

\begin{proof}
  The proof is similar to the one of Lemma~\ref{lem:epred-lin-noiseless}.
  Starting at \eqref{eq:epred-lin-ol} and considering only noise-affected open-loop data, at the ideal parameters the error becomes
  \begin{equation}
    \elind = - \left[ \Aqpd \Qdiq - \Bqpd \Cfq \right]\Vt
  \end{equation}
  Then, using \eqref{eq:aq-inv} gives \eqref{eq:epred-lin-noisy}.

  On the other hand, if noisy data are collected during a closed-loop experiment with an initial controller $\Coq$, starting at \eqref{eq:epred-lin-cl}, the ideal parameters make the error
  \begin{align}
    \elind &= -\left[ \Aqpd \Qdiq - \Bqpd \Cfq \right] \Soq \Vt
        + \Aqpd \Coq  \Soq \Vt \nonumber\\
      &= - \Aqpd \left[ \Gqi + \Coq \right]\Soq \Vt \label{eq:epred-lin-5},
  \end{align}
  where \eqref{eq:epred-lin-5} comes from \eqref{eq:aq-inv}.
  Then \eqref{eq:epred-lin-noisy} is proven by using the followin equality in \eqref{eq:epred-lin-5}:
  \begin{equation}
    \left[ G\inv(q) +  C(q) \right]S(q) = G\inv(q),\label{eq:ginv}
  \end{equation}
  which can be obtained straight from \eqref{eq:qqp} and \eqref{eq:q-inv}.
\end{proof}

\begin{theorem}
  Under Assumptions~\ref{ass:matching} and \ref{ass:noise}, and considering the affine prediction error \eqref{eq:epred-lin} obtained from noisy data, the limit of the cost \eqref{eq:v-corr} when $L \rightarrow \infty$ and $N \rightarrow \infty$ has a minimum at the ideal parameters.
\end{theorem}
\begin{proof}
  Lemma~\ref{lem:epred-lin-noisy} shows that the error at the ideal parameters is simply filtered noise, which is uncorrelated with the experimental input by Assumption~\ref{ass:noise}; therefore, $\Vcd$ is a minimum of \eqref{eq:v-corr}.
\end{proof}

\subsubsection{Filter design}

When Assumption \ref{ass:matching} does not hold, the minima of the disturbance response \eqref{eq:v-dr} and correlation costs \eqref{eq:v-corr} may differ, as happened with the prediction error cost function.
In this case, the filter may be designed to bring the two minima close together.

\begin{theorem}
  When using noiseless data and the linear predictor in \eqref{eq:u-lin}, the filter that force  the correlation cost \eqref{eq:v-corr} to be the same as the disturbance response cost \eqref{eq:v-dr} for a given disturbance signal $\Dt$ is one such that
  \begin{align}
    \abs*{\Kwp}^2 &= \frac{\abs*{\Qwp \Qdw}^2\Phid}{\abs*{\Awp}^2\abs*{\Phixy}^2},\label{eq:filter-corr-lin-1}
  \end{align}
  where $\Phixy$ is the Fourier transform of the correlation, i.e. it is the cross power spectrum of the experiment input --- either $\Ut$ or $\Rt$ --- and the output $\Yt$.
\end{theorem}
\begin{proof}
  Considering the noiseless case, using Parseval's theorem and Lemma~\ref{lem:epred-lin-noiseless}, the frequency-domain version of \eqref{eq:v-corr} may be written as
  \begin{equation}
    \lim_{\substack{N\to\infty\\L\to\infty}} \Vc = \frac{1}{2\pi}\int_{-\pi}^{\pi}
      \abs*{\Kwp}^2 \abs*{\Awp}^2 \abs*{\Qiwp - \Qdiw}^2
      \abs*{\Phixy}^2 \dd \omega.\label{eq:v-corr-lin-parseval}
  \end{equation}
  Compare \eqref{eq:v-corr-lin-parseval} with the frequency-domain version of the disturbance response cost \eqref{eq:v-dr-parseval}.
  Equating \eqref{eq:v-dr-parseval} and \eqref{eq:v-corr-lin-parseval} and solving for the filter shows that the two costs are the same if the filter is given by \eqref{eq:filter-corr-lin-1}.
\end{proof}

\begin{remark}
  Observe that the filter \eqref{eq:filter-corr-lin-1} may be approximated by
  \begin{equation}
    \abs*{\Kwp}^2 \approx \frac{\abs*{\Qdw \Qdw}^2\Phid}{\abs*{\Awp}^2\abs*{\Phixy}^2},\label{eq:filter-corr-lin}
  \end{equation}
  when $\Qwp \approx \Qdw$.
\end{remark}

\begin{remark}
  All quantities involved in \eqref{eq:filter-corr-lin} are available or can be estimated from data, but the filter depends on the parameters even after the approximation.
  One possibility to circumvent this issue is using an iterative optimization where a least squares problem is solved at each iteration.
  Then, use $\Aqp = 1$ in the filter's denominator at the first iteration, and the filter's denominator is updated after each iteration, as before.
\end{remark}

In the noiseless case, the squared cross power spectrum may also be written
\begin{equation}
  \abs*{\Phi_{u,y}(\omega)}^2 = \abs*{\Gw}^2 \Phi_u^2(\omega),
\end{equation}
for open-loop data, and
\begin{equation}
  \abs*{\Phi_{r,y}(\omega)}^2 = \abs*{\Tow}^2 \Phi_r^2(\omega),
\end{equation}
for closed-loop data.
These equalities may be employed to simplify the filter when it is possible to use the disturbance that will be rejected as the probing signal, i.e. $\Ut = \Dt$ in open-loop or $\Rt$ = $\Dt$ in closed-loop.

\subsection{Non-linear predictor}

Now, consider what happens when the non-linear predictor \eqref{eq:u-nlin} is employed with the correlation approach.
Again, first considering noiseless data.

\begin{theorem}
  Under Assumptions~\ref{ass:matching} and \ref{ass:noise}, and considering the non-linear prediction error \eqref{eq:epred-nlin} obtained from noiseless data, the limit of the cost \eqref{eq:v-corr} when $L \rightarrow \infty$ and $N \rightarrow \infty$ has a global minimum at the ideal parameters.
\end{theorem}

\begin{proof}
  With noiseless data, from Lemma~\ref{lem:epred-nlin-noiseless}, the prediction error becomes zero at the ideal parameters, regardless if the data are from open-loop or closed-loop experiment. Therefore, for noiseless data, the correlation between the input and the prediction error is also zero and $\Vcd = 0$ is a global minimum of \eqref{eq:v-corr}.
\end{proof}

Observe that when using \eqref{eq:epred-nlin} the cost \eqref{eq:v-corr} is non-convex; therefore, an iterative optimization must be performed.
As before, an initial guess for the parameters is required and needs to be close enough to the minimum for the estimate to converge to the correct parameters.

Now, consider what happens when using the non-linear predictor obtained from noise-affected data.

\begin{lemma}\label{lem:epred-nlin-noisy}
  When using noise-affected data, the prediction error obtained with the non-linear predictor \eqref{eq:u-nlin} at the ideal parameters becomes simply
  \begin{equation}
    \enld = -\Gqi \Vt,\label{eq:epred-nlin-noisy}
  \end{equation}
  regardless if the data are collected in open-loop or closed-loop.
\end{lemma}
\begin{proof}
  First, consider only open-loop data.
  In this case, the prediction error at the ideal parameters becomes
  \begin{align}
    \enld &= - \left[ \Qdiq - \Cdq \right]v(t),
  \end{align}
  from \eqref{eq:epred-nlin-ol}, using the ideal parameters.
  Then, using \eqref{eq:c-d} results in \eqref{eq:epred-nlin-noisy}.

  Now consider the case where the noisy data are collected during a closed-loop experiment with an initial controller $C_0(q)$.
  In this case, starting at \eqref{eq:epred-nlin-cl}, the prediction error at the ideal parameters becomes
  \begin{align}
    \enld &= - \left[ \Qdiq - \Cdq + \Coq \right]\Soq \Vt = - \left[ \Gqi + \Coq \right] \Soq \Vt \label{eq:epred-nlin-4}
  \end{align}
  where \eqref{eq:epred-nlin-4} comes from \eqref{eq:q-inv}.
  Then, using \eqref{eq:ginv} in \eqref{eq:epred-nlin-4} yields \eqref{eq:epred-nlin-noisy}.
\end{proof}

\begin{theorem}
  Under Assumptions~\ref{ass:matching} and~\ref{ass:noise}, and considering the non-linear predictor \eqref{eq:u-nlin}, the cost \eqref{eq:v-corr} has a minimum at the ideal parameters.
\end{theorem}
\begin{proof}
  Lemma~\ref{lem:epred-nlin-noisy} shows that the error at the ideal parameters becomes simply filtered noise, which is uncorrelated with the experiment input by Assumption~\ref{ass:noise}; therefore, $\Vcd$ is a minimum of the cost \eqref{eq:v-corr}.
\end{proof}

\subsubsection{Filter design}

As before, the filter may be designed to make the minimum of \eqref{eq:v-corr} close to the minimum of \eqref{eq:v-dr}.
\begin{theorem}
  When using noiseless data and the non-linear predictor \eqref{eq:u-nlin}, the filter that will force the correlation cost \eqref{eq:v-corr} to be the same as the disturbance response cost \eqref{eq:v-dr} for a given disturbance signal $\Dt$ is one such that
  \begin{equation}
    \abs{\Kwp}^2 = \frac{\abs*{\Qwp \Qdw}^2\Phid}{\abs*{\Phixy}^2},\label{eq:filter-corr-nlin-1}
  \end{equation}
\end{theorem}
\begin{proof}
  Considering the noiseless case, using Parseval's theorem and Lemma~\ref{lem:epred-nlin-noiseless}, the frequency-domain version of \eqref{eq:v-corr} may be written as
  \begin{equation}
    \lim_{\substack{N\to\infty\\L\to\infty}}\Vc = \frac{1}{2\pi}\int_{-\pi}^{\pi}
      \abs*{\Kwp}^2 \abs*{\Qiwp - \Qdiw}^2
      \times \abs*{\Phixy}^2 \dd \omega. \label{eq:v-corr-nlin-parseval}
  \end{equation}

  Now, compare the cost \eqref{eq:v-corr-nlin-parseval} with the frequency-domain version of the disturbance response cost in \eqref{eq:v-dr-parseval}.
  Equating \eqref{eq:v-corr-nlin-parseval} and \eqref{eq:v-dr-parseval} and solving for the filter indicates that \eqref{eq:filter-corr-nlin-1} is makes the two functions the same.
\end{proof}

\begin{remark}
  Once more, the filter in \eqref{eq:filter-corr-nlin-1} may be approximated by one such that
  \begin{equation}
  \abs{\Kwp}^2 \approx \frac{\abs*{\Qdw \Qdw}^2\Phid}{\abs*{\Phixy}^2},\label{eq:filter-corr-nlin}
  \end{equation}
  when $\Qwp \approx \Qdw$.
\end{remark}

\begin{remark}
  Once again, all the quantities involved in \eqref{eq:filter-corr-nlin} are available or may be estimated from data.
  Also, the approximate filter does not depend on the parameters.
  Unfortunately, an iterative optimization is still needed because the cost function is non-convex.
\end{remark}

\section{Examples}
\label{sec:examples}

This section presents three simulation case studies to exemplify how the proposed methodology may be employed to tune the parameters of linear feedback controllers aiming at disturbance rejection.
The system simulated is the following, arbitrarily selected,
\begin{equation}
  \Gq = \frac{1}{120}\frac{(1-0.7\z^{-1})}{(1-0.95\z^{-1})^2},\label{eq:exp-g}
\end{equation}
which is stable and has unitary static gain.
The control objective is to reject step disturbances in steady-state, which may be accomplished by including an integrator in the controller structure.
We recall that within the data-driven framework, the process model \eqref{eq:exp-g} is supposed to be unknown by the user; therefore, it is presented here for the sake of the example.
The model is needed only to generate the simulated data and to present what would be the ideal and optimal responses for comparison with the other responses obtained.

The theory presented indicates that a broad range of controllers may be tuned by these methods.
However, for these examples, we have chosen controllers from the PID family because these are a simple, powerful, and widespread type of controllers.
A PIDF (PID with filtered derivative), for example, presents the following structure:
\begin{equation}
  C^\text{PIDF}(\z,\p) = \frac{b_0 + b_1 \z^{-1} + b_2 \z^{-2}}{1 + a_1 z^{-1}}\frac{1}{1-\z^{-1}},\label{eq:pidf-exp}
\end{equation}
where the controller's fixed part comprises the integrator, while the other part has four parameters that must be identified.
A simpler PI controller, on the other hand, presents the following structure:
\begin{equation}
  C^\text{PI}(\z,\p) = \frac{b_0 + b_1 \z^{-1}}{1}\frac{1}{1-\z^{-1}},\label{eq:pi-exp}
\end{equation}
with the same fixed part.
Note that any of the parts could happen to be non-causal, depending on the choices made, but this is not a problem because all the calculations are performed offline.
However, the complete controller, must be causal, otherwise it will not be implementable.

The following reference model for the desired load disturbance sensitivity is chosen
\begin{equation}
  \Qdq = \frac{1}{120}\frac{(1-0.7\z^{-1})(1-\z^{-1})}{(1-0.9\z^{-1})(1-0.95\z^{-1})^2},\label{eq:exp-qd}
\end{equation}
which is about as fast as the open-loop process \eqref{eq:exp-g} and rejects constant disturbances in steady-state.
% To ensure the controller's causality, the reference model's step response has the same value for the first sample as the open-loop process step response.
Within this scenario, the ideal controller, obtained from \eqref{eq:c-d}, would be
\begin{equation}
  \Cdq = \frac{12(1 - 0.95\z^{-1})^2}{(1 - 0.7\z^{-1})(1 - \z^{-1})},\label{eq:exp-cd}
\end{equation}
which is a PIDF controller and can be represented with the controller structure in \eqref{eq:pidf-exp}.

If the controller structure available is PIDF, then Assumption~\ref{ass:matching} holds and $C(\z,\pd) = \Cdq$, otherwise the methods presented yield a controller that is not optimal in terms of the cost function \eqref{eq:v-dr}.
In this case, the controller estimated may deliver performance far from the desired one.
However, as stated before, a filter may be employed to estimate a controller closer to the optimum.
These three cases are explored in the following sections.

\subsection{Matching case}

The first case study considers that Assumption~\ref{ass:matching} holds, which means a PID with derivative filter is used and $C(\z,\pd) = \Cdq$.
To generate the data, 100 Monte Carlo simulations are performed in open-loop and another 100 simulations in closed-loop, varying only the measurement noise realization $\Vt$.
The simulations in closed-loop use the controller $\Coq = C^\text{PIDF}(\z,0.5\pd)$.
During the simulations, the experiment input signal comprises 3000 samples of a square wave with levels $\pm 1$ and period 300 samples.
The measurement noise is a gaussian white process with variance $\sigma^2 = 0.0025$.
For illustrative purposes, the first period of one realization of the signals collected in open-loop is presented in Figure~\ref{fig:exp-ol}.
\begin{SCfigure}
  \includegraphics{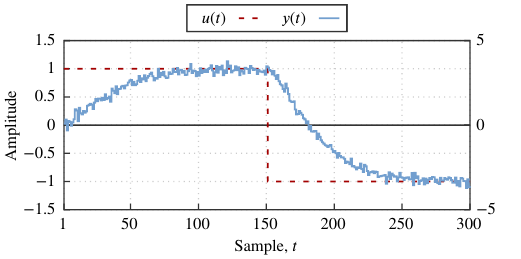}
  \caption{\label{fig:exp-ol} First period in one realization of the open-loop signals collected.
  In this figure, $\Ut$ is the probing square wave (red dashed), $\Yt$ is the noisy output signal (light blue).}
\end{SCfigure}
The first period of one realization of the closed-loop signals is also presented in Figure~\ref{fig:exp-matching-cl}, for comparison.
\begin{SCfigure}
  \includegraphics{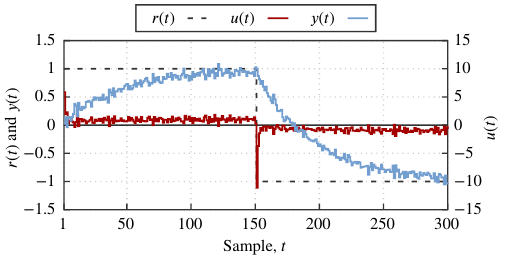}
  \caption{\label{fig:exp-matching-cl}First period in one realization of the closed-loop signals collected with $\Coq$ a PIDF controller.
  In this figure, $\Rt$ is the probing square wave (black dashed), $\Ut$ is the control action (solid red), and $\Yt$ is the noisy output signal (light blue).}
\end{SCfigure}
Observe that in this case, the two signals collected along with the probing signal are affected by the noise entering the system by the process output.
In both the cases, the controller's virtual input is affected by the noise, which means this is an errors-in-variables identification problem.

The prediction error is constructed using the collected data either through \eqref{eq:epred-lin}, for the linear predictor, or through \eqref{eq:epred-nlin}, for the non-linear one, both using $n_a = 1$ and $n_b = 3$.
When Assumption~\ref{ass:matching} holds and there is no noise, the ideal controller is achieved for any filter $\Kqp$.
Since there is noise in this simulation, a filter $\Kqp = \Qdq$ was employed to reduce high frequency noise.

First, 400 controllers are estimated by minimizing the norm of the prediction error \eqref{eq:v-pe}, one controller for each realization, type of experiment (open-loop and closed-loop) and type of predictor (linear and non-linear).
When using a linear prediction, the controller is computed analytically and using the non-linear predictor, 1000 iterations of the \verb|fminsearch| algorithm are performed, initialized at $\Coq$.
The controllers identified are employed to simulate the resulting closed-loop systems, $\Qqp$.
Finally the following approximate of the disturbance response cost function \eqref{eq:v-dr} is calculated for each controller using those simulated systems:
\begin{equation}
  \hat V_\text{dr}(\p) = \frac{1}{150} \sum_{t=1}^{150} \norm*{\left[ \Qdq - \Qqp \right]\Dt }^2,\label{eq:exp-cost}
\end{equation}
where $\Dt$ is step signal.
The resulting values for those controllers are presented through the boxplots in Figure~\ref{fig:boxplot-matching-vdft}.
\begin{SCfigure}
  \includegraphics{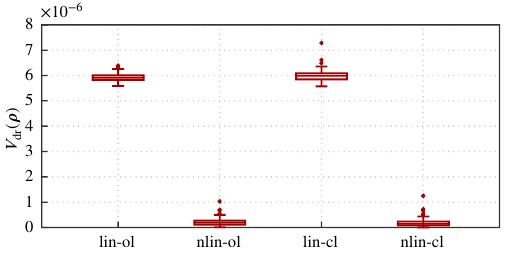}
  \caption{\label{fig:boxplot-matching-vdft} Boxplots showing the distribution of the disturbance response costs achieved by minimizing the prediction error norm when Assumption~\ref{ass:matching} held. Results are grouped by type of predictor: linear (lin) and non-linear (nlin), and type of experiment: open-loop (ol) and closed-loop (cl).}
\end{SCfigure}
Observe that all controllers achieved a very good performance, with a very small cost (oreder of $10^{-6}$).
Notice that there is no significant difference between controllers computed using open-loop or closed-loop data.
Also, note that controllers computed using the non-linear predictor achieved a better performance than the controllers computed using the linear predictor.
If there was no noise in the signals, all methods would achieve the ideal controller, so the non-linear predictor is less sensitive to noise.

The correlation algorithm that minimize the correlation cost \eqref{eq:v-corr} was also used to identify other 400 controllers, one for each realization, type of experiment and predictor.
Again, the controller is computed analytically when using the linear predictor and when using the non-linear predictor 1000 iterations of the \verb|fminsearch| algorithm are performed, starting at $\Coq$.
The filter is still the same and the number of lags is fixed at $L = 185$.
The resulting values of the cost function for these controllers are presented through the boxplots in Figure~\ref{fig:boxplot-matching-dcbt}.
\begin{SCfigure}
  \includegraphics{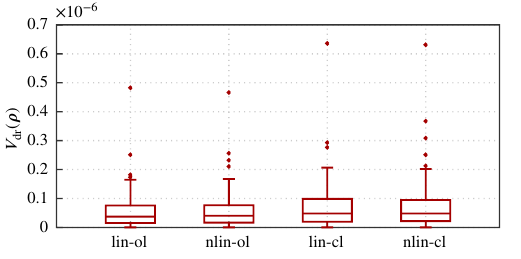}
  \caption{\label{fig:boxplot-matching-dcbt} Boxplots showing the distribution of the disturbance response costs achieved by minimizing the correlation between the prediction error and the probing signal. Results are grouped by type of predictor: linear (lin) and non-linear (nlin), and type of experiment: open-loop (ol) and closed-loop (cl).}
\end{SCfigure}
Observe that the correlation method is less sensitive to noise and there is no significant difference between the linear and non-linear predictors nor between data collected in open-loop or closed-loop.
Table \ref{tab:exp} shows a comparison between the mean and variance of the cost function $\hat V_\text{dr}(\p)$ for each method.
Again, if there was no noise in the signals, all methods would achieve the ideal controller, but the correlation approach is much less sensitive to noise and achieve smaller costs than the method that minimize the norm of the prediction error.

In mostl real-life situations Assumption~\ref{ass:matching} cannot be guaranteed to hold and in those cases the filter $\Kqp$ may be employed to reduce the distance between the minimum of the data-driven cost function and the actual cost that we want to minimize, as shown in the next examples.

\subsection{Mismatching case}

Until now, we have considered the case where the available controller structure matches the one of the ideal controller.
Now, consider the more realistic case where the controller that must be tuned has a smaller controller structure.
For the sake of the example we suppose this is a PI controller such as \eqref{eq:pi-exp}; therefore, lacking the derivative structure.
Since the ideal controller does not match the controller structure available, different disturbances will give different performances and we must choose which disturbance response we want to optimize.
We have chosen the step response, i.e. $\Dt$ is a step as before.
For completeness, we have numerically obtained the optimum controller
\begin{equation}
  C(\z,\pstar) = \frac{4.1381(1 - 0.9788\z^{-1})}{(1 - \z^{-1})},\label{eq:exp-c-opt}
\end{equation}
by minimizing the disturbance response cost \eqref{eq:exp-cost} directly.
Once again, note that in a real-life scenario this would not be possible because the process model is not identified.
Observe that $\Cdq \neq C(\z,\pd)$; therefore, some performance loss is expected with the resulting tuned controllers.

The algorithm that minimizes the prediction error's norm is used to identify $400$ controllers, one for each realization, type of experiment and type of predictor.
The filter employed is estimated as \eqref{eq:filter-pe-lin} (linear predictor) or \eqref{eq:filter-pe-nlin} (non-linear predictor).
Observe that since there is no parameter to identify on the denominator of the controller, both linear and non-linear predictors should achieve the same performance.
In the case of the linear predictor, the controller is computed directly where in the case of the non-linear predictor, 1000 iterations of the \verb|fminsearch| algorithm are performed starting at $\Coq$.
The costs estimated by \eqref{eq:exp-cost} for each controller identified are presented in Figure~\ref{fig:boxplot-filter-vdft}, grouped by experiment and predictor.
\begin{SCfigure}
  \includegraphics{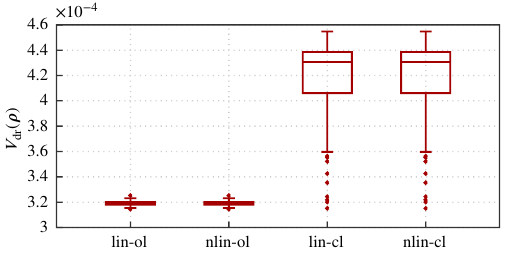}
  \caption{\label{fig:boxplot-filter-vdft}Boxplots showing the distribution of the disturbance response costs achieved by minimizing the norm of the prediction error, filtered to compensate for the violation of Assumption~\ref{ass:matching}. Results are grouped by type of predictor: linear (lin) and non-linear (nlin), and type of experiment: open-loop (ol) and closed-loop (cl).}
\end{SCfigure}
Observe that the linear and non-linear predictors achieve exactly the same performance as expected.
Observe also that the performance is worse using the PI controllers than the PIDF controllers because the reference model cannot be achieved exactly using the reduced order controller.
Notice that the controllers computed using the open-loop data achieved a better performance than the controllers computed using the closed-loop data.
Ideally (without noise), there should not be difference between these two methods since the filters \eqref{eq:filter-pe-lin} and \eqref{eq:filter-pe-nlin} should compensate the frequency differences between signals.
However, the filters are designed for noiseless case and a performance loss is expected when there is noise affecting the signals.

Then, other 400 controllers are identified using the same data by minimizing the correlation between the prediction error and the experiment input, one for each realization, type of experiment and predictor.
The number of lags employed is $L = 185$ and  the filter is estimated either through \eqref{eq:filter-corr-lin} (linear predictor) or \eqref{eq:filter-corr-nlin} (non-linear predictor).
When using the linear predictor, the controllers are computed directly and when using the non-linear predictor, 1000 iterations of the \verb|fminsearch| algorithm are performed starting at $\Coq$.
The costs estimated by \eqref{eq:exp-cost} for each controller are presented in the boxplots of Figure~\ref{fig:boxplot-filter-dcbt}, grouped by experiment and predictor.
\begin{SCfigure}
  \includegraphics{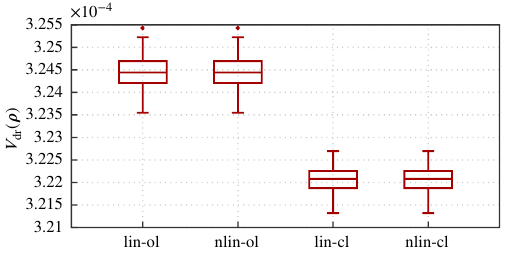}
  \caption{\label{fig:boxplot-filter-dcbt}Boxplots showing the distribution of the disturbance response costs achieved by minimizing the correlation between the prediction error, filtered to compensate for the viola of Assumption~\ref{ass:matching}, and the probing signal. Results are grouped by type of predictor: linear (lin) and non-linear (nlin), and type of experiment: open-loop (ol) and closed-loop (cl).}
\end{SCfigure}
Again the linear and non-linear predictors achieve exactly the same performance as expected, that  is worse using the PI controllers than the PIDF controllers.
Notice that now the controllers computed using the open-loop data achieved a worse performance than the controllers computed using the closed-loop data, but the difference between them is very small.
Table \ref{tab:exp} shows that for open-loop data the best performance was obtained minimizing the norm of the prediction error method while for closed-loop data the best performance was obtained with the correlation approach.
%
%
%We notice that the costs obtained are still small, but in this case the filter actually made the costs increase when compared with the ones obtained with the unfiltered data from Figure~\ref{fig:boxplot-mismatching-dcbt}.
%This is due mostly to two factors combined.
%Firstly, the filter is designed considering noiseless signals, such that when there is noise, the filter may amplify it and increase the variance of the parameters.
%Secondly, the optimal filter matches those designed only if the spectra are well estimate, which is difficult to achieve because the data are finite and there is a significant amount of noise.
%
%In conclusion, when Assumption~\ref{ass:matching} does not hold and the noise is significant, there are times that the filter may reduce the disturbance response cost, but this is not guaranteed.
%On the other hand, when the noise contributions are light, using the filter is probably a good idea, since it will force the identified controller closer to the optimum.
%It is possible that a filter could be designed to force the two minima closer while attenuating the noise effects, but this is still an open issue.
%
%\subsection{Overall comparison}
%
%The mean and standard deviations for the disturbance response cost function values obtained at each case are presented in Table~\ref{tab:exp} for comparison.
\newcommand{\mc}{\multicolumn}
\newcommand{\mr}{\multirow}
\begin{table}
  \centering
  \caption{\label{tab:exp}Mean and standard deviation of the costs.}
  \begin{tabular}{llcccc}
  \toprule
    &   & \mc{2}{c}{\sc open-loop} & \mc{2}{c}{\sc closed-loop} \\
  \midrule
    &   & linear & non-linear & linear & non-linear \\
  \midrule
  \mc{6}{l}{\sc matching, no filter ($\times 10^{-6}$)} \\
    & norm\hspace{4.0cm} & $5.9147$ $(\pm0.1524)$ & $0.2154$ $(\pm0.1606)$ & $5.9931$ $(\pm0.2328)$ & $0.1959$ $(\pm0.1799)$ \\
    & correlation        & $0.0564$ $(\pm0.0661)$ & $0.0578$ $(\pm0.0670)$ & $0.0722$ $(\pm0.0849)$ & $0.0750$ $(\pm0.0889)$ \\
  \midrule
%  \mc{6}{l}{\sc mismatching, no filter ($\times 10^{-4}$)} \\
 %   & norm         & $6.2171$ $(\pm0.1342)$ & $6.2171$ $(\pm0.1342)$ & $5.4154$ $(\pm0.0900)$ & $5.4154$ $(\pm0.0900)$ \\
  %  & correlation  & $3.2129$ $(\pm0.0041)$ & $3.2129$ $(\pm0.0041)$ & $3.2163$ $(\pm0.0032)$ & $3.2163$ $(\pm0.0032)$ \\
  %\midrule
  \mc{6}{l}{\sc mismatching, filter ($\times 10^{-4}$)} \\
    & norm         & $3.1899$ $(\pm0.0188)$ & $3.1899$ $(\pm0.0188)$ & $4.1635$ $(\pm0.3353)$ & $4.1635$ $(\pm0.3353)$ \\
    & correlation  & $3.2448$ $(\pm0.0037)$ & $3.2448$ $(\pm0.0037)$ & $3.2207$ $(\pm0.0031)$ & $3.2207$ $(\pm0.0031)$ \\
  \bottomrule
\end{tabular}

\end{table}
%When Assumption~\ref{ass:matching} holds and minimizing the norm of the prediction error, the non-linear predictor clearly outperformed the linear one, although the costs obtained with the latter were already very small.
%In all the other situations the results obtained using both predictors are practically the same within a given case.
%
%As expected, the best results are obtained when Assumption~\ref{ass:matching} holds regardless if optimizing the prediction error or the correlation.
%When that assumption does not hold, the filter may help to find the best controller, specially when optimizing the norm of the prediction error.
%Unfortunately, the filter may also increase the value of the objective function because of how it is estimated and because it is calculated for noiseless data.
%However we expect it to be helpful in most situations.
%
%When considering the same case, the correlation approach yields results approximately as good or much better than the ones obtained with the norm of the prediction error.
%If the noise level were increased this behaviour should become even more evident.
%
%Finally, regarding open versus closed-loop data, the correlation approach is less sensitive of the experiment setup, while the prediction error approach suffers more.
%This is specially true in terms of the variance of the cost in the last case (mismatched controller, filtered data), when the correlation approach behaves better.

\section{Conclusion}
\label{sec:conclusion}

This article presented a system identification framework for data-driven controller tuning, aiming at disturbance rejection using and reference model approach.
The main idea of this work is to obtain virtual signals from experiments in open-loop or closed-loop, and then to estimate the parameters of a restricted order controller solving an optimization problem that does not depend on a process model.
It is proposed the use of two different one-step ahead predictors: one linear and another non-linear. 
Also, it is proposed two different approaches to estimate the parameters of the controller.
The first one minimizes the quadratic norm of the predicton error while the second minimizes the correlation between the prediction error and an external signal.
When restricted order controllers are employed such that the prescribed reference model can not be achieved, an extra filter is proposed to reduce the difference between the obtained and prescribed load disturbance sensitivities.
Simulations showed that proposed methods achieve very good performance both for full order and reduced order controllers.

\section*{Acknowledgements}
This study was financed in part by the Coordenação de Aperfeiçoamento de Pessoal de Nível Superior - Brasil (CAPES) - Finance Code 001.
This work was supported in part by the Conselho Nacional de Desenvolvimento Científico e Tecnológico (CNPq).

\bibliography{ref}

\end{document}